\documentclass{l4dc2022}

\title[Learning for networked estimation]{Learning distributed channel access policies  for networked estimation: data-driven optimization in the mean-field regime}
\usepackage{times}
\usepackage[capitalize]{cleveref}
\usepackage{algorithmic}
\usepackage{algorithm}
\usepackage{wrapfig}
\usepackage{lipsum}
\usepackage{enumerate}
\usepackage[shortlabels]{enumitem}
\usepackage{comment}
\usepackage{hyperref}

\author{%
 \Name{Marcos M. Vasconcelos} \Email{marcosv@vt.edu}\\
 \addr Virginia Polytechnic Institute and State University
}

\usepackage{lipsum}

\DeclareMathOperator{\Equaldef}{\overset{def}{=}}

\begin{document}

\maketitle

\thispagestyle{empty}

\begin{abstract}
The problem of communicating sensor measurements over shared networks is  prevalent in many modern large-scale distributed systems such as cyber-physical systems, wireless sensor networks and the internet of things. Due to bandwidth constraints, the system designer must jointly design decentralized medium access transmission and estimation policies that accommodate a very large number of devices in extremely contested environments such that the collection of all observations is reproduced at the destination with the best possible fidelity. We formulate a remote estimation problem in the mean-field regime where a very large number of sensors communicate their observations to an access point, or base-station, under a strict constraint on the maximum fraction of transmitting devices. We show that in the mean-field regime, this problem exhibits a structure which enables tractable optimization algorithms. More importantly, we obtain a data-driven learning scheme that admits a finite sample-complexity guarantee on the performance of the resulting estimation system under minimal assumptions on the data's probability density function.
\end{abstract}


\section{Introduction}

State estimation is a fundamental building block of control systems subject to uncertain disturbances. Most modern applications rely on smart sensors that collect and transmit data to a remote location. In particular, for cyber-physical systems, where a large number of sensors collect and transmit data over a bandwidth limited network, there is a need to communicate efficiently and maintain the quality of the estimates as close as possible to optimal. Moreover, these conflicting goals must be achieved under the lack of knowledge of the underlying stochastic model of the environment. This task requires a combination of techniques in machine learning, decision-theory, and optimization. 


Consider a distributed remote estimation system, where a very large number of sensors collect data and communicate with a destination over a constrained network. 
To obtain a tractable problem, we assume that the data is independent and identically distributed across sensors. When the number of sensors in the system tends to infinity, this assumption leads to an optimization problem for which the optimal communication policy is determined by designing a quantizer under a rate constraint, a problem known to be non-convex. We show how to obtain a locally optimal solution using an algorithm known as the Convex-Concave procedure, which is guaranteed to converge regardless of the underlying probability density function of the data.

During the last decade, there has been significant research activity in systems with an asymptotically large number of coupled decision-making agents \citep{Mahajan:2013,Gagrani:2020,Sanjari:2021}. Such problem formulations are very relevant in modern applications such as industrial internet of things, where many tiny, low-power devices sense the environment and communicate with a remote gateway, base-station or access point \citep{Gatsis:2021}. Another application where the mean-field regime is appropriate for analysis is in robotic swarms \citep{Zheng:2021}. 
State estimation for control of power systems in the mean-field regime has been considered in the literature \citep[and references therein]{Chen:2017}. The main idea behind studying teams with an infinite number of identical/symmetric agents is to allow the system designer to perform the optimization over a single policy instead of accounting for the overwhelming combinatorial-like complexity that arise from the interactions among finitely many discrete agents.

The problem of remote estimation over a collision channel was introduced in \citep{Vasconcelos:2017} for a finite number of sensors communicating over a channel with unit capacity. Recently, \cite{Zhang:2021} extended the setup for a channel with arbitrary capacity, and introduced the first instance of such remote estimation problems when the distribution is unknown. Finally, \cite{Vasconcelos:2021b} considered using a nonparametric statistics technique known as \textit{Kernel Density Estimation} (KDE) \citep{Wasserman:2006}, to obtain a scheme that is more stable and and sample-efficient than the one obatained via empirical-risk minimization, also known as sample average approximation in the stochastic programming literature \citep{Shapiro:2021}. We use the fact that our algorithm is independent of the data's pdf to come up with a data-driven scheme where the system designer relies only on a finite collection of data points using KDE. KDE enables a data-based algorithm and with the choice of a Gaussian kernel, which leads to an algorithm is free of integrals, implying in great numeric stability \citep{Asi:2019}. More importantly, we obtain a finite sample complexity bound on the probability that system will operate within the desired collision free regime, as a function of the state-dimension, the number of samples and a capacity back-off parameter. Our sample complexity result reveals an information theoretic interpretation, where the number of samples required to guarantee operation in the target regime is inversely proportional to how far we are willing to operate from the system capacity.

\begin{figure}[t!]
\centering
\includegraphics[width=\textwidth]{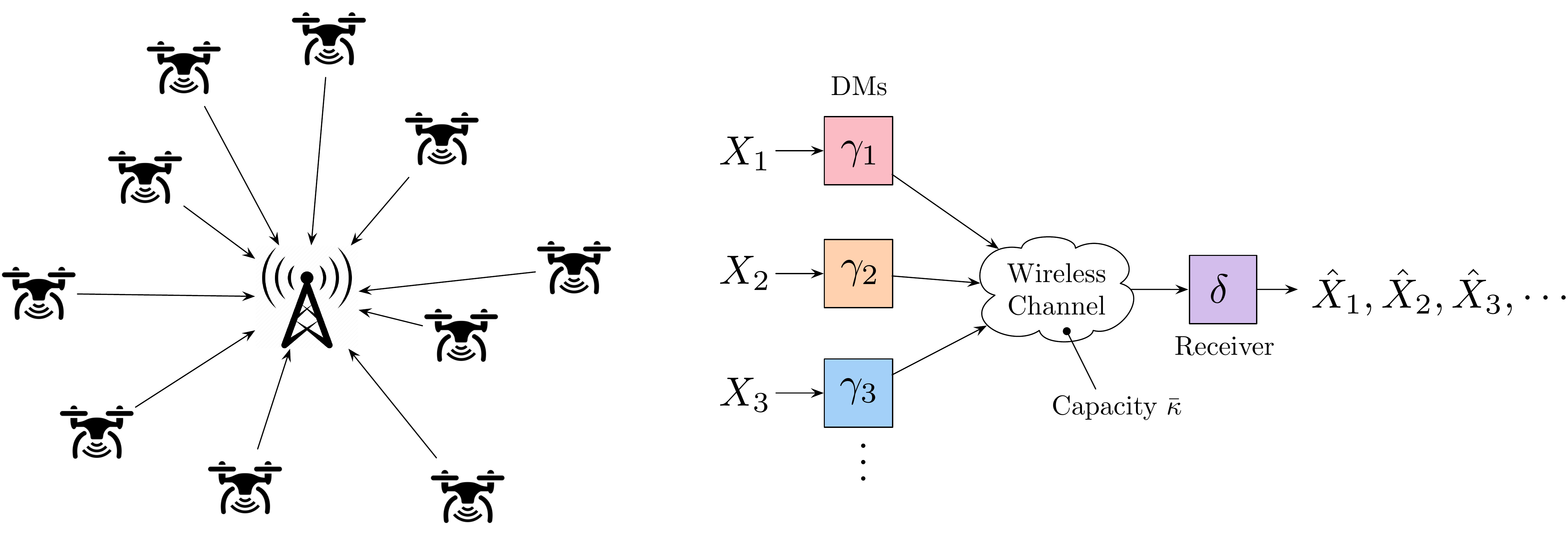}
\caption{System diagram: A large number of decision makers in a remote estimation system communicating their observations to a base station or access point without coordination, and its block diagram abstraction considered for analytical problem formulation herein.}
\end{figure}


\section{Problem setup}

Consider a remote estimation system where a collection of $n$ agents observe independent and identically distributed (iid) random vectors $\{X_i\}_{i=1}^n$, where $X_i \in \mathbb{R}^d$ and $X_i \sim f$, $i\in \{1,\cdots,n\}$. Assume throughout this paper that the data's probability density function (pdf) $f$ is smooth, and supported over the entire $\mathbb{R}^d$.\footnote{For instance, this assumption is satisfied if a Gaussian mixture describes the data well and encompasses a broad class of distributions.} The agents communicate their observations to a remote destination over a shared communication medium: only a portion $\kappa(n)<n$ of the agents is allowed to communicate noiselessly at any time. The agents behave strategically as to which observations to transmit. After observing $X_i$, the $i$-th agent makes a binary decision $U_i\in\{0,1\}$ to transmit its observation ($U_i=1$) or to remain silent ($U_i=0$). The decision variable $U_i$ is computed according to a policy $\gamma_i:\mathbb{R}^d \rightarrow \{0,1\}$, such that $U_i=\gamma_i(X_i).$

If the number of transmissions is less than or equal to the capacity $\kappa(n)$, all the transmitted observations are perfectly recovered at the receiver, and only the non-received obervations need to be estimated. If the number of transmissions is larger than the channel capacity $\kappa(n)$, nothing can be decoded at the receiver, and \textit{all} of the observations need to be estimated. This is mathematically described by the following estimation rule:
\begin{equation}\label{eq:estimator}
\hat{X}_i = \begin{cases}
X_i, &  \text{if} \ \ \sum_{j=1}^nU_j \leq \kappa(n), \ U_i=1 \\ 
\theta_i, &  \text{if} \ \ \sum_{j=1}^nU_j \leq \kappa(n), \ U_i=0 \\ 
\alpha_i, &  \text{if} \ \ \sum_{j=1}^nU_j > \kappa(n)
\end{cases}
\end{equation}
where $\theta_i,\alpha_i \in \mathbb{R}^d$, $i\in \{1,\cdots,n\}$. Our goal is to jointly design the communication policies $\boldsymbol{\gamma}\Equaldef \{\gamma_i\}_{i=1}^n$ and the estimation parameters $\boldsymbol{\theta}\Equaldef \{\theta_i\}_{i=1}^n$ and $\boldsymbol{\alpha}\Equaldef \{\alpha_i\}_{i=1}^n$ that minimize the normalized mean-squared error (NMSE) between the observations and their estimates:
\begin{equation}\label{eq:nmse}
\mathcal{J}_n(\boldsymbol{\gamma},\boldsymbol{\theta},\boldsymbol{\alpha}) \Equaldef \frac{1}{n} \mathbf{E} \left[ \sum_{i=1}^n\big\|X_i-\hat{X}_i\big\|^2 \right].
\end{equation}

\begin{theorem}
If $\lim_{n\rightarrow \infty} \kappa(n)/n = \bar{\kappa} \in (0,1)$, the problem of jointly optimizing the NMSE with respect to the transmission and estimation policies for a system where $n\rightarrow \infty$ is equivalent to solving:
\begin{align}\label{eq:problem}
\underset{\gamma \in \Gamma, \ \theta \in \mathbb{R}^d}{\mathrm{minimize}} \qquad & \mathbf{E}\big[\|X-\theta\|^2 \mid U=0\big]\mathbf{P}\big(U=0 \big)  \\ 
\mathrm{subject \ to}\qquad & \mathbf{P}\big(U=1 \big) \leq \bar{\kappa},
\end{align}
where $X\sim f$ and $\Gamma = \big\{ \gamma \mid \gamma: \mathbb{R}^d \rightarrow \{0,1\}\big\}$. 
\end{theorem}





\begin{proof}
Using the law of total expectation, we expand the NMSE in \cref{eq:nmse} as:
\begin{multline}
\mathcal{J}_n(\boldsymbol{\gamma},\boldsymbol{\theta},\boldsymbol{\alpha}) 
 =  \frac{1}{n} \sum_{i=1}^n \mathbf{E} \left[ \big\|X_i-\hat{X}_i\big\|^2 \ \bigg|\ U_i=0,\sum_{j\neq i}^n U_j \leq \kappa(n)\right]\mathbf{P}\left(U_i=0\right) \mathbf{P}\left( \sum_{j=1}^n U_j\leq \kappa(n) \right) \\
+ \frac{1}{n} \sum_{i=1}^n \mathbf{E} \left[ \big\|X_i-\hat{X}_i\big\|^2 \ \bigg|\ U_i=1,\sum_{j\neq i}^n U_j \leq \kappa(n)-1\right]\mathbf{P}\left(U_i=1\right) \mathbf{P}\left( \sum_{j\neq i}^n U_j\leq \kappa(n) -1 \right)
  \\ +
 \frac{1}{n} \mathbf{E} \left[ \sum_{i=1}^n\big\|X_i-\hat{X}_i\big\|^2 \ \bigg|\ \sum_{j=1}^n U_j > \kappa(n)\right]\mathbf{P}\left( \sum_{j=1}^n U_j > \kappa(n) \right).
\end{multline}
Then, using the definition of the estimator in \cref{eq:estimator}, and the iid assumption, we get: 
\begin{multline}
\mathcal{J}_n(\boldsymbol{\gamma},\boldsymbol{\theta},\boldsymbol{\alpha}) 
 =  \frac{1}{n} \sum_{i=1}^n \mathbf{E} \left[ \big\|X_i-\theta_i\big\|^2 \ \bigg|\ U_i=0\right]\mathbf{P}\left(U_i=0\right) \mathbf{P}\left( \sum_{j=1}^n U_j\leq \kappa(n) \right) 
  \\ +
 \frac{1}{n} \mathbf{E} \left[ \sum_{i=1}^n\big\|X_i-\alpha_i\big\|^2 \ \bigg|\ \sum_{j=1}^n U_j > \kappa(n)\right]\mathbf{P}\left( \sum_{j=1}^n U_j > \kappa(n) \right).
\end{multline}

From the iid assumption of the observation sequence $\{X_i\}_{i=1}^n$, we may assume without loss of optimality that every agent uses the same policy $\gamma_i = \gamma, \ i\in\{1,\cdots,n\}$. Thus, the number of agents that decide to communicate their observations to the receiver is a Binomial random variable. Let $p = \mathbf{P}(U_i=1)$. 
For $\bar{\kappa} \in (0,1)$, the following holds:
 \begin{equation}\label{eq:convergence}
 \lim_{n\rightarrow \infty} \frac{\kappa(n)}{n} = \bar{\kappa} \ \ \ \Longrightarrow \ \ \ \mathbf{P}\Bigg(\sum_{j=1}^n U_j > \kappa(n)\Bigg) \longrightarrow \mathbf{1}(p>\bar{\kappa}), \ \ \mathrm{a.e.}
\end{equation}
The last step is proved using the following sequence of identities:
\begin{multline}
\lim_{n\rightarrow \infty} \mathbf{P}\Bigg(\sum_{j=1}^n U_j > \kappa(n)\Bigg) =  \lim_{n\rightarrow \infty} \mathbf{E} \left[ \mathbf{1}\left(\frac{1}{n}\sum_{j=1}^nU_j > \frac{\kappa(n)}{n}\right)\right] \\
 \stackrel{(a)}{=}  \mathbf{E} \left[ \lim_{n\rightarrow \infty}  \mathbf{1}\left(\frac{1}{n}\sum_{j=1}^nU_j > \frac{\kappa(n)}{n}\right)\right] 
 \stackrel{(b)}{=}  \mathbf{E} \Big[  \mathbf{1}\big(p > \bar{\kappa}\big)\Big]  \stackrel{(c)}{=}  \mathbf{1}\big(p > \bar{\kappa}\big),
\end{multline}
where $(a)$ follows from the Lebesgue dominated convergence theorem, $(b)$ follows from the strong law of large numbers, and $(c)$ follows from the fact that the argument of the expectation operator is a deterministic function. Therefore, using the fact that $\{X_i\}$ is iid, we argue that $\theta_1=\theta_2=\cdots$. Moreover, it is possible to show that the optimal value of $\alpha^\star_i$ satisfies:
\begin{equation}
\lim_{n\rightarrow \infty} \alpha^\star_i = \mathbf{E}\bigg[X_i  \  \bigg|  \  \sum_{j=1}^nU_j>\kappa(n)\bigg] = \mathbf{E}[X], \ \ i\in\{1,2,\cdots\}.
\end{equation}


\noindent Observe from \cref{eq:convergence} that by an appropriate choice of $p$, we may operate in one of two regimes: the \textit{collision-free} regime ($p\leq \bar{\kappa}$) or the \textit{constant-collision} regime ($p>\bar{\kappa}$). And consequently, the NMSE in the regime when $n\rightarrow \infty$ becomes:
\begin{equation}\label{eq:Jinfty}
\mathcal{J}_{\infty} (\gamma,\theta) = \begin{cases}
\mathbf{E}\big[\|X-\theta\|^2 \mid U=0\big]\mathbf{P}\big(U=0 \big)  & \text{if} \ \ \mathbf{P}\big(U=1 \big) \leq \bar{\kappa} \\
\mathbf{E}\big[\|X- \mathbf{E}[X]\|^2 \big] & \text{otherwise.}
\end{cases}
\end{equation}
Clearly, the constant-collision regime offers worse performance than the collision-free regime. 
\end{proof}


\begin{remark}
Notice that in the asymptotic regime, the optimization variables $\boldsymbol{\alpha}$ are not present, meaning that the detrimental effect of collisions can be mitigated by adjusting the probability of transmission to operate below the asymptotic channel capacity parameter $\bar{\kappa}$.
\end{remark}


\section{Solving the optimization problem}

Once the constrained optimization problem for the remote estimation system in the asymptotic regime is established, we may obtain the structure of its optimal solutions. Notice that there are two variables in the optimization problem:  the function $\gamma$, and a vector $\theta$. Therefore, this is an infinite dimensional problem. Moreover, the objective in the optimization problem is non-convex. 
Writing the Lagrangian associated with \cref{eq:problem}, we characterize its saddle points. Let $\lambda \geq 0$ and define:
\begin{equation}\label{eq:Lagrangian}
\mathcal{L}(\gamma,\theta,\lambda) \Equaldef \int_{\mathbb{R}^d} \|x-\theta\|^2 \mathbf{1}\big(\gamma(x) = 0 \big) f(x)dx + \lambda\Big(\int_{\mathbb{R}^d}  \mathbf{1}\big(\gamma(x) = 1 \big) f(x)dx -\bar{\kappa} \Big)
\end{equation}
First, we have the following structural result, which establishes that the communication policy in every saddle point is of the threshold type \citep[and references therein]{Vasconcelos:2017}.

\begin{proposition} \label{prop:threshold}
If $(\theta^\star, \gamma^\star, \lambda^\star)$ is a saddle-point of $\mathcal{L}(\theta,\gamma,\lambda)$, then
\begin{equation}\label{eq:threshold}
\gamma^\star(x) = \begin{cases}
0, & \text{if} \ \ \|x-\theta^\star\|^2 \leq \lambda^\star \\
1, & \text{otherwise.}
\end{cases}
\end{equation}
Furthermore, 
\begin{equation}\label{eq:probability}
\mathbf{P} \big(\|X-\theta^\star\|^2 > \lambda^\star \big) = \bar{\kappa}
\ \ \
and
 \ \ \
 \theta^{\star} = \frac{1}{1-\bar{\kappa}}\int_{\mathbb{R}^d} x\mathbf{1}\big(\| x-\theta^{\star}\|^2 \leq \lambda^\star  \big) f(x)dx.
\end{equation} 
\end{proposition}

\begin{proof}
From the definition of the Lagrangian in \cref{eq:Lagrangian}, for any $\theta \in \mathbb{R}^d$ and $\lambda \geq 0$, the communication policy that minimizes it may be constructed by assigning to a set $\mathcal{A}_0 \subset \mathbb{R}^d$, all the points $x\in \mathbb{R}^d$ such that $\| x-\theta\|^2\leq \lambda$. Otherwise, the point is assigned to $\mathcal{A}_1 = \mathbb{R}^d\setminus\mathcal{A}_0$. If there are any sets with non-zero measure for which $\| x-\theta\|^2\leq \lambda$ assigned to $\mathcal{A}_1$, we can reduce the Lagrangian even further and $(\theta^\star, \gamma^\star, \lambda^\star)$ would not constitute a saddle-point. 

To establish the second result, we take the derivative of \cref{eq:Lagrangian} with respect to $\lambda$ and use \cref{eq:threshold}, which leads to:
\begin{eqnarray}
\frac{\partial}{\partial \lambda}\mathcal{L}(\gamma^\star,\theta^\star,\lambda^\star) = 0 &\Longrightarrow & \int_{\mathbb{R}^d} \mathbf{1}\big(\gamma^{\star}(x) = 1 \big) f(x)dx = \mathbf{P}\Big(\|X -\theta^\star \|^2> \lambda^\star \Big) = \bar{\kappa}.
\end{eqnarray}

Lastly, taking the gradient of \cref{eq:Lagrangian} with respect to $\theta$, we have:
\begin{eqnarray}
\nabla_{\theta}\mathcal{L}(\gamma^\star,\theta^\star,\lambda^\star) = 0 &\Longrightarrow & \int_{\mathbb{R}^d} -2(x-\theta^\star)\mathbf{1}\big( \gamma^\star(x) = 0 \big) f(x)dx =0. 
\end{eqnarray}
Then, using the characterization in \cref{eq:threshold,eq:probability}, we have:
\begin{equation}\label{eq:centroid}
\int_{\mathbb{R}^d}x\mathbf{1}\big( \|x-\theta^\star \|^2 \leq \lambda^\star \big) f(x)dx = \theta^\star \underbrace{\int_{\mathbb{R}^d}\mathbf{1}\big( \|x-\theta^\star \|^2 \leq \lambda^\star \big) f(x)dx}_{= 1-\bar{\kappa}}.
\end{equation}
\end{proof}



The structural results stated in Proposition~\ref{prop:threshold} help us define an algorithm that converges to a saddle-point of the Lagrangian and therefore defines a locally optimal policy for the remote estimation system in the asymptotic regime. First, consider the Laplacian when $\lambda$ and $\theta$, are arbitrarily fixed. Then, using \cref{eq:threshold}, we obtain the function $\tilde{\mathcal{L}}: \mathbb{R}^{d+1} \rightarrow \mathbb{R}$ defined as:
\begin{equation}\label{eq:mod_laplacian}
\tilde{\mathcal{L}}(\theta,\lambda) \Equaldef \mathcal{L}(\gamma^\star,\theta,\lambda) = \mathbf{E}\Big[\min\big\{ \| X-\theta\|^2,\lambda \big\}\Big] - \lambda\bar{\kappa}.
\end{equation}

We are now interested in minimizing the function $\tilde{\mathcal{L}}$ over $\theta$ for any $\lambda\geq 0$. The function in \cref{eq:mod_laplacian} is non-convex in $\theta$. In fact, this new objective function is the expectation of a clipped convex function, which is known to lead to NP hard problems for discrete random variables \citep{Barratt:2020}. Therefore, there is little hope that we will be able to obtain globally optimal solutions without making additional assumptions on the problem's structure. However, the function possess a very convenient structure, namely, that the term in the expectation can be decomposed as a difference of convex functions, also known as a \textit{DC decomposition}. The DC decomposition enables the use of an algorithm known as the \textit{Convex-Concave Procedure} (CCP) \citep{Lipp:2016}. For the purpose of optimization over $\theta$, we may temporarily ignore the term $-\lambda\bar{\kappa}$.


\subsection{Convex-Concave Procedure}
To obtain a DC decomposition of the objective function, simply notice that the expectation of the minimum can be expressed as follows:
\begin{equation}
\mathbf{E}\Big[\min\big\{ \| X-\theta\|^2,\lambda \big\}\Big] = \mathbf{E}\Big[\|X-\theta\|^2+\lambda\Big] - \mathbf{E}\Big[\max\{\|X-\theta\|^2,\lambda\}\Big].
\end{equation}

Let $\mathcal{G}_{\lambda}(\theta)\Equaldef\mathbf{E}\big[\max\big\{\|X-\theta\|^2,\lambda \big\}\big]$. The CCP consists of approximating the objective function by replacing $\mathcal{G}_{\lambda}(\theta)$ with its affine approximation at a point $\theta_k \in \mathbb{R}^d$, then minimizing the resulting approximate convex function. Conveniently, the decomposition above leads to:  
\begin{equation}
\theta_{k+1} = \arg \min_{\theta \in \mathbb{R}^d} \Big\{ \mathbf{E}\big[\|X-\theta\|^2\big]+\lambda - \Big( \mathcal{G}_\lambda(\theta_k) + g_\lambda(\theta_k)^T(\theta-\theta_k) \Big)  \Big\},
\end{equation}
where $g_{\lambda}$ is \textit{any} subgradient of $\mathcal{G}_\lambda$. The first benefit of using the DC decomposition above where the first function is quadratic, is that solving the optimization problem above is equivalent to finding the solution of the first-order optimality condition 
\begin{equation}
\nabla_{\theta} \tilde{\mathcal{L}}(\theta_{k+1},\lambda) = - \mathbf{E}\big[2(X-\theta_{k+1})\big] - g_\lambda(\theta_k) =0 \Longrightarrow \theta_{k+1} = \frac{1}{2}g_\lambda(\theta_k)+\mathbf{E}[X].
\end{equation}


The second benefit of using the DC decomposition stems from the fact that computing a subgradient $g_{\lambda}$ of the expectation of pointwise maximum of functions is extremely simple. For instance, using elementary subgradient calculus, the following is a suitable choice:
\begin{equation}
g_\lambda(\theta) = -2\mathbf{E}\Big[\big(X-\theta\big)\mathbf{1}\big(\|X-\theta\|^2>\lambda\big)\Big],
\end{equation}
which leads to:
\begin{equation}\label{eq:CCP}
\theta_{k+1} = \mathbf{E}\Big[X\mathbf{1}\big(\|X-\theta_k\|^2\leq\lambda\big)\Big] + \theta_{k}\mathbf{P}\big(\|X-\theta_k\|^2>\lambda\big).
\end{equation}

The recursion in \cref{eq:CCP} is guaranteed to converge \citep{Lanckriet:2009,Lipp:2016}. Moreover, this convergence is irrespective of the pdf $f$ of the random vector $X$. Then, by considering the limit of the sequence $\{\theta_k\}\rightarrow \theta^\star$, we have:
\begin{equation}
\theta^\star = \frac{\mathbf{E}\Big[X\mathbf{1}\big(\|X-\theta^\star\|^2\leq\lambda\big)\Big]}{1-\mathbf{P}\big(\|X-\theta^\star\|^2 > \lambda\big)}.
\end{equation}
If $\lambda$ is chosen such that $\mathbf{P}\big(\|X-\theta^\star\|^2 > \lambda\big) =\bar{\kappa}$, the CCP in \cref{eq:CCP} converges to a saddle point of the Lagrangian and defines a locally optimal solution of the optimization problem in the asymptotic regime. Algorithm \ref{alg:CCP} summarizes the discussion above.

\begin{algorithm}[t]

\caption{Data-driven policy optimization for remote estimation in the mean-field regime}
\begin{enumerate}
\item Given a batch of data $\mathcal{D}=\{x_m\}_{m=1}^M$, estimate $\hat{f}_M$. If $f$ is known, $\hat{f}_M \gets f$
\item \textbf{Initialization}: Compute:
\begin{equation*}
\theta_0 \gets \int_{\mathbb{R}^d} x\hat{f}_M(x)dx
\end{equation*}
and solve for $\lambda^\star$ the following equation:
\begin{equation*}
\int_{\mathbb{R}^d} \mathbf{1}(\|x -\theta_0\|^2\leq \lambda^\star) \hat{f}_M(x)dx = 1 - \bar{\kappa}
\end{equation*}
Set $\lambda_0 \gets \lambda^\star.$

\item \textbf{Iteration}:  For $k=0, 1, 2,\cdots$

\begin{itemize}[leftmargin = 4mm]

\item[a.] Repeat the following recursion until convergence to $\theta^\star$

\begin{equation*}
\theta_{\ell+1} \gets \int_{\mathbb{R}^d} x \mathbf{1}(\|x -\theta_\ell\|^2\leq \lambda_k) \hat{f}_M(x)dx - \bar{\kappa}\theta_\ell
\end{equation*}
Set $\theta_{k+1} \gets \theta^\star$

\item[b.] Solve for $\lambda^\star$ the following equation

\begin{equation*}
\int_{\mathbb{R}^d} \mathbf{1}(\|x - \theta_k\|^2\leq \lambda^\star) \hat{f}_{M}(x)dx = 1 - \bar{\kappa}
\end{equation*}
Set $\lambda_{k+1} \gets \lambda^\star$

\item[c.] Repeat until convergence

\end{itemize}
\end{enumerate}
\label{alg:CCP}
\end{algorithm}

\section{Data-driven algorithm}
One remarkable feature of this problem is the number of structural and convergence results we can obtain without the explicit knowledge of the pdf $f$. However, our algorithm relies on the availability of a smooth pdf to compute the sequence of thresholds $\lambda_k$. For that reason, a data-driven version of our algorithm cannot rely on non-smooth approximations, e.g. using the empirical distribution of the data. Our problem requires a smooth data-driven approximation. One alternative is to use Kernel Density Estimation (KDE), where the pdf $f$ is approximated by:
\begin{equation}
\hat{f}_{M}(x)=\frac{1}{M}\sum_{m=1}^M \prod_{j=1}^d\frac{1}{\sqrt{2\pi}h_M}\exp\left(-\frac{(x_j-x_{mj})^2}{2h_M} \right) \ \ \text{with} \ \ h_M = \frac{1.06}{M^{1/5}}\min\left\{s,\frac{Q}{1.34}\right\},
\end{equation}
which corresponds to KDE using a product Gaussian kernel \citep{Wasserman:2006}. 
The quantity $M$ is the number of data samples in a batch, $s$ is the sample standard deviation and $Q$ is the interquantile range\footnote{The interquantile range is the difference between the 75th and 25th percentiles of the data.}. In addition to the smoothness of the estimate $\hat{f}_M$, one special feature of using KDE to approximate the pdf is that the estimation quality can be quantified via the mean-squared integration error (MISE):
\begin{equation}\label{eq:MISE}
\mathrm{MISE} = \mathbf{E}\left[\int_{\mathbb{R}^d}\big|\hat{f}_{M}(x)-f(x) \big|^2dx\right]= \mathcal{O}\Big(M^{-4/(d+4)}\Big),
\end{equation}
where the expectation is taken with respect to the true distribution $f$ from which the data samples $\{x_m \}_{m=1}^M$ are drawn. It can be shown that no other non-parametric density estimation scheme can produce estimates with a better sample-complexity if the estimation quality is measured using the MISE criterion \citep{Vaart:2000}.

Since the algorithm described in the previous section is agnostic to the pdf $f$, a data-driven version of it simply uses a KDE $\hat{f}_M$ instead of $f$. Following \cite{Wasserman:2006}, we have $\hat{f}_M \stackrel{P}{\rightarrow} f$, $M\rightarrow \infty$. Therefore, \Cref{alg:CCP} produces good estimates of locally optimal solutions to \cref{eq:problem} as the size of the data-set increases. In the next section, we provide a sample complexity result to obtain good data-driven locally optimal solutions to \cref{eq:problem} in the absence of knowledge of the true probabilistic model $f$.

\section{Sample complexity analysis}

Consider a random batch of data samples $\{X_m\}_{m=1}^M$. Assuming that $(\hat{\theta}_M,\hat{\lambda}_M)$ are computed using the procedure in Algorithm 1, they are considered in this section as random variables. Thus, using the approximate pdf to obtain a channel access policy parametrized by $(\hat{\theta}_M,\hat{\lambda}_M)$ may yield a true probability of channel access that violates the constraint in \cref{eq:problem}. Such violation may be very detrimental. Here, we would like to provide a finite sample guarantee that the true probability of channel access when using a policy computed from sampled data will not violate the channel capacity constraint with high-probability. 
Let $\mathbf{P}(U=1) = 1 - \int_{\hat{\theta}_M-\sqrt{\hat{\lambda}_M}}^{\hat{\theta}_M+\sqrt{\hat{\lambda}_M}} f(x)dx,$ and notice that this quantity is itself a random variable.



\begin{theorem}[Finite sample complexity guarantee] Consider a remote estimation system designed according to Algorithm 1. Assume that the pdf of the observations $f(x)$ is smooth with full support on $\mathbb{R}^d$. For a wireless network with assymptotic capacity $\bar{\kappa}$ and a channel access policy designed to operate with probability of transmission $\bar{\kappa}-\delta$, for $\delta\in(0,\bar{\kappa})$. Then, the true probability that the the capacity constraint will be violated decays to zero with a rate that satisfies:
\begin{equation}
\mathbf{P}\Big( \mathbf{P}(U=1)>\bar{\kappa} \Big) \leq \mathcal{O}\left(\frac{1}{\delta M^{2/(d+4)}}\right).
\end{equation}
\end{theorem}

\begin{proof}
To guarantee that the probability constraint will be met with high probability, we need to back off from the capacity by a pre-specified constant $\delta\in(0,\bar{\kappa})$. 
Let $\hat{\mathbf{P}}(U=1) \Equaldef  1 - \int_{\hat{\theta}_M-\sqrt{\hat{\lambda}_M}}^{\hat{\theta}_M+\sqrt{\hat{\lambda}_M}} \hat{f}_M(x)dx = \bar{\kappa}$ (by design). Then, notice that 
\begin{equation}
\mathbf{P}\left( \hat{\mathbf{P}}(U=1)-\mathbf{P}(U=1) \geq \delta \right) \leq \mathbf{P}\left( \big|\hat{\mathbf{P}}(U=1)-\mathbf{P}(U=1) \big|\geq \delta \right).
\end{equation}


The following inequalities hold:
\begin{eqnarray*}
\Big|\mathbf{P}(U=1)-\hat{\mathbf{P}}(U=1)\Big| & = & \left|\int_{\hat{\theta}_M-\sqrt{\hat{\lambda}_M}}^{\hat{\theta}_M-\sqrt{\hat{\lambda}_M}} \Big[\hat{f}_M(x)-f(x)\Big]dx\right| \\
& \stackrel{(a)}{\leq} & \int_{\mathbb{R}} \mathbf{1}\Big({\hat{\theta}_M-\sqrt{\hat{\lambda}_M}} \leq x \leq {\hat{\theta}_M-\sqrt{\hat{\lambda}_M}} \Big) \left|\hat{f}_M(x)-f(x)\right|dx \\
&\stackrel{(b)}{\leq}  & \left(2\sqrt{\hat{\lambda}_M}\right)^{1/2} \left(\int_{\mathbb{R}}\left|\hat{f}_M(x)-f(x)\right|^2dx\right)^{1/2}
\end{eqnarray*}
where $(a)$ follows from the Triangle inequality and $(b)$ follows from H\"{o}lder's inequality. Finally, we apply the expectation operator with respect to the data $\{X_m\}_{m=1}^M$, obtaining via the Cauchy-Scharz inequality the following bound:
\begin{equation}
\mathbf{E}\left[\Big|\mathbf{P}(U=1)-\hat{\mathbf{P}}(U=1)\Big|\right] \leq \left(2\mathbf{E}\left[\sqrt{\hat{\lambda}_M}\right]\right)^{1/2} \left(\mathbf{E}\left[\int_{\mathbb{R}^d}\left|\hat{f}_M(x)-f(x)\right|^2dx\right]\right)^{1/2}.
\end{equation}

Now, since we obtain $\hat{\lambda}_M$ from an algorithmic procedure that depends on a batch of sampled data, we have very little control over it, but we can guarantee that this term is bounded for any number of samples $M$. 
Therefore, using \cref{eq:MISE}, we obtain:
$\mathbf{E}\left[\Big|\mathbf{P}(U=1)-\hat{\mathbf{P}}(U=1)\Big|\right] = \mathcal{O}\big(M^{-2/(d+4)}\big).$ 
The following holds:
\begin{equation}
\mathbf{P}\big( \mathbf{P}(U=1) > \bar{\kappa} \big)  = 
\mathbf{P}\big(\mathbf{P}(U=1) - \hat{\mathbf{P}}(U=1)> \delta \big) \\
 \stackrel{(c)}{\leq}  \frac{\mathbf{E}\big[ \big| \mathbf{P}(U=1) - \hat{\mathbf{P}}(U=1) \big|  \big]}{\delta},
\end{equation}
where $(c)$ follows from Markov's inequality. 
\end{proof}

\begin{remark}[Interpretation]
The result above mimics results related to the capacity in the Shannon sense. To guarantee a decay in the probability of violating the channel capacity constraint we back off from $\bar{\kappa}$ by $\delta$. Notice that as $\delta \rightarrow 0$, our scheme based on KDE requires an increasingly large number of samples. 
\end{remark}

\section{Numerical results}
Consider data samples generated according to the following Gaussian mixture:
\begin{equation}
X \sim 0.2\mathcal{N}(-2,0.2)  + 0.2\mathcal{N}(-1,0.075)   + 0.1\mathcal{N}(0,0.1)   + 0.3\mathcal{N}(1,0.1)  + 0.2\mathcal{N}(2,0.1). 
\end{equation}

\begin{wrapfigure}{R}{0.6\textwidth}
\centering
\includegraphics[width=0.6\textwidth]{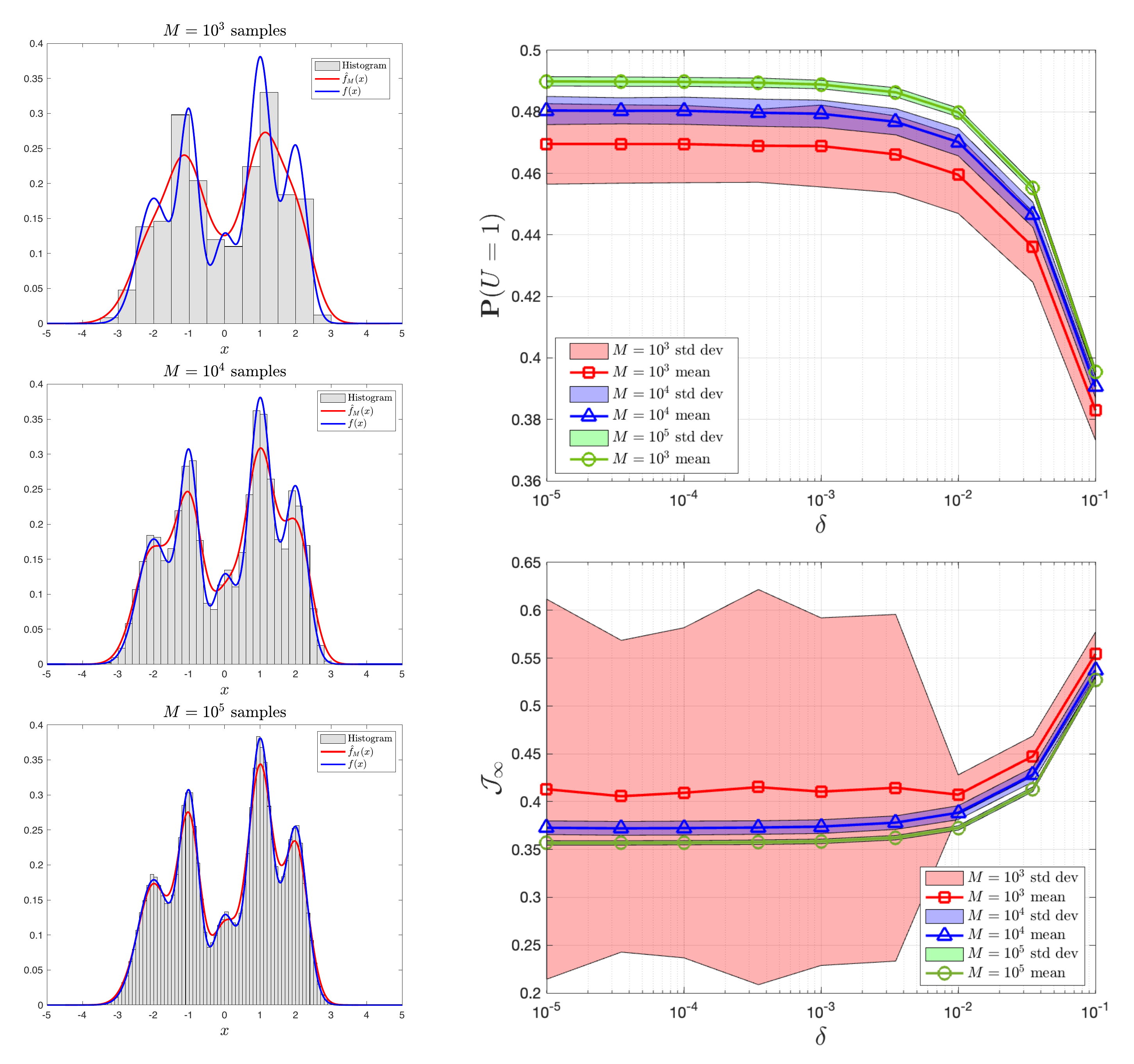}
\vspace{-20pt}
\caption{Numerical results.}
\label{fig:numerical}
\end{wrapfigure}
Assuming that this distribution is known, and setting $\bar{\kappa}=0.5$, \cref{alg:CCP} results in $\theta^\star=0.0592$ and $\lambda^\star = 1.5063.$ The value of the optimal NMSE in the asymptotic regime is $\mathcal{J}_{\infty}^\star = 0.3411$. Using the data-driven approach based on KDE for several different sample sizes ($M=10^3,10^4,10^5$), we obtain the curves displayed in \cref{fig:numerical}, where we vary the capacity back-off parameter from $\delta=10^{-5}$ to $10^{-1}$. The estimation performance of the system degrades with $\delta$, but the probability that the capacity constraint will be violated decreases sharply to zero even for moderate values of the sample size \cref{fig:numerical} (right, top). However, the KDE is very sensitive to $M$, which can be seen from \cref{fig:numerical} (left), yielding in a large variance in performance for smaller values of $M$, e.g. $M=10^3$ \cref{fig:numerical} (right, bottom). An important observation is that the performance for systems designed using a random batch of data, quickly concentrate around their means, e.g. $M>10^4$.\footnote{All the code used to obtain the numerical results contained herein can be found on \url{https://github.com/mullervasconcelos/L4DC22.git}}

\section{Conclusions and future work}

Scalability issues often challenge modern communication networks and their medium access control mechanisms. We argue that when it comes to applications such as the Internet of Things or in robotic swarms, it is beneficial from the analytical perspective to study the design of policies in the mean-field regime, i.e., systems with an infinite number of devices. One problem of interest consists of estimating an iid process where each random variable is accessed by a different device and transmitted over a constrained network to a remote location. Here, we showed that the interference between simulataneously transmitting agents can be effectively mitigated with appropriate design. More importantly, if we back off from the channel capacity by a constant amount, we can guarantee that the data-driven design based on Kernel Desnity Estimation will perform satisfactorily, meaning, will not violate a probability constraint with high probability.




\newpage

\bibliographystyle{natbib}
\bibliography{L4DC22}

\end{document}